\documentclass[10pt]{IEEEtran}

\usepackage{amsmath,amssymb,amsthm}
\usepackage{epsfig,graphicx}
\usepackage[svgnames,rgb]{xcolor}
\usepackage[color=yellow]{pdfcomment}
\usepackage{colortbl,hhline}
\usepackage{algorithmic,algorithm,lipsum}
\usepackage{graphicx}
\usepackage{caption}
\usepackage{subcaption}

\usepackage[mathscr]{eucal}
\usepackage{amsbsy}
\usepackage{bm}
\usepackage{fixltx2e}
\MakeRobust{\overrightarrow}

\newcommand{\Sec}[1]{\hyperref[sec:#1]{\S\ref*{sec:#1}}} 
\newcommand{\Eqn}[1]{\hyperref[eq:#1]{(\ref*{eq:#1})}} 
\newcommand{\Fig}[1]{\hyperref[fig:#1]{Figure~\ref*{fig:#1}}} 
\newcommand{\Tab}[1]{\hyperref[tab:#1]{Table~\ref*{tab:#1}}} 
\newcommand{\Thm}[1]{\hyperref[thm:#1]{Theorem~\ref*{thm:#1}}} 
\newcommand{\Lem}[1]{\hyperref[lem:#1]{Lemma~\ref*{lem:#1}}} 
\newcommand{\Prop}[1]{\hyperref[prop:#1]{Property~\ref*{prop:#1}}} 
\newcommand{\Cor}[1]{\hyperref[cor:#1]{Corollary~\ref*{cor:#1}}} 
\newcommand{\Def}[1]{\hyperref[def:#1]{Definition~\ref*{def:#1}}} 
\newcommand{\Alg}[1]{\hyperref[alg:#1]{Algorithm~\ref*{alg:#1}}} 
\newcommand{\Ex}[1]{\hyperref[ex:#1]{Example~\ref*{ex:#1}}} 

\newcommand{\V}[1]{{\bm{\mathbf{\MakeLowercase{#1}}}}} 
\newcommand{\mb}[1]{\mathbb{#1}}
\newcommand{\mc}[1]{\mathcal{#1}}

\newcommand{\M}[1]{{\bm{\mathbf{\MakeUppercase{#1}}}}} 

\newcommand{\qil}{\M{Q}_{i}^{(\ell)}}
\newcommand{\qilt}{\M{Q}_{i}^{(\ell)\top}}

\title{On deterministic conditions for subspace clustering under missing data}
\author{Wenqi Wang, Shuchin Aeron and Vaneet Aggarwal}

\newtheorem{theorem}{Theorem}[section]
\newtheorem{lemma}[theorem]{Lemma}

\newcommand{\oil}{\Omega_{i}^{(\ell)}}
\newcommand{\oi}{\Omega_i}

\newcommand{\oji}{\Omega_{i,j}}

\newcommand{\nuil}{\V{\nu}_{i}^{(\ell)}}
\newcommand{\lamil}{\V{\lambda}_{i}^{(\ell)}}
\newcommand{\uell}{\M{U}^{(\ell)}}

\newcommand{\vell}{\M{V}^{(\ell)}}
\newcommand{\xell}{\M{X}^{(\ell)}}
\newcommand{\Aell}{\M{A}^{(\ell)}}
\newcommand{\tAell}{\tilde{\M{A}}^{(\ell)}}
\newcommand{\aell}{\V{a}^{(\ell)}}

\begin{document}
\pagenumbering{gobble}

\maketitle

\begin{abstract}
In this paper we present deterministic analysis of sufficient conditions for sparse subspace clustering under missing data, when data is assumed to come from a Union of Subspaces (UoS) model. In this context we consider two cases, namely Case I when all the points are sampled at the same co-ordinates, and Case II when points are sampled at different locations. We show that results for Case I directly follow from several existing results in the literature, while results for Case II are not as straightforward and we provide a set of dual conditions under which, perfect clustering holds true. 
We provide extensive set of simulation results for clustering as well as completion of data under missing entries, under the UoS model. Our experimental results indicate that in contrast to the full data case, accurate clustering does not imply accurate subspace identification and completion, indicating the natural order of relative hardness of these problems.
\end{abstract}

\section{Introduction}
In this paper we consider the problem of data clustering under the union of subspaces (UOS) model \cite{soltanolkotabi2012,Elhamifar:2012uz}, when each data vector is sampled in an element-wise manner. This is referred to as the case of \emph{missing data}. In other words we are looking to harvest a union of subspaces structure from the data, when the data is missing. Such a problem has been recently considered in a number of papers \cite{bal2,bal3,ErikssonArXiv2011,Yang2015}. This setting has implications to data completion under the union of subspaces model in contrast to the single subspace model that has been prevalent in the matrix completion literature. In contrast to statistical analysis in \cite{bal2,bal3,ErikssonArXiv2011}, this paper uses a variant of the sparse subspace clustering (SSC) algorithm \cite{Elhamifar:2012uz} to give sufficient deterministic conditions for accurate subspace clustering under missing data. In contrast to \cite{Yang2015}, which does not provide any specific conditions for success of SSC under missing data, in this paper we provide implications of the deterministic conditions for several specific cases of sampling. Further through extensive simulations we demonstrate for the first time that accurate clustering under missing data does not imply accurate subspace clustering and completion thereby indicating the natural order of hardness of these problems under missing data.



\section{Problem set-up}
\label{sec:prob_setup}

We are given a set of data points collected as columns of a matrix $\M{X}$, i.e. $\M{X}_{i}, i = 1,2,...,N$ such that, $\M{X}_i \in \bigcup_{\ell =1}^{L} \mb{S}^{(\ell)}$, where $\mb{S}^{(\ell)}$ is a subspace of dimension $d_\ell$ in $\mb{R}^{d}$, for $\ell = 1,2,...,L$. Further, each data point $\M{X}_i$ is sampled at $\Omega_i$ co-ordinates (randomly or deterministically) and the problem is to identify the subspaces $\mb{S}^{(\ell)}$ under missing data. In order to derive meaningful performance guarantees for the proposed algorithm, we consider the following generative model for the data. Let $\xell$ denote the set of vectors in $\M{X}$ which belong to subspace $\ell$. Let $$ \xell = \uell \Aell ,$$ where the $N_\ell$ columns of $\Aell$ are drawn from the unit sphere $\mc{S}^{d_\ell-1}$ and $\uell$ is a matrix with orthonormal columns, whose columns span the subspace, $\mb{S}^{(\ell)}$.  Under missing data, point $\xell_i$ is sampled at locations $\oil$ and 
\begin{align} \xell_{\oi} = \mathbf{I}_{\oil} \uell \V{a}_{i}^{(\ell)}
\end{align}
where $\mathbf{I}_{\oil}$ is a diagonal matrix with $\mathbf{I}_{\oil}(k,k) = 1$ iff $ k \in \oil$. It is essentially a zero filled $\M{X}_{i}^{(\ell)}$.




For the sake of exposition, in the following we will assume that $d_\ell = d$ for all $\ell = 1,2,...,L$.

Given this set-up the problem is to accurately cluster the data points such that within each cluster the data points belong to the same (original) subspace. 
\vspace{-3mm}
\subsection{The Algorithm: SSC-LP}
\label{sec:alg}
Our algorithm as presented below is a minor variation of the SSC-EWZF algorithm in \cite{Yang2015} in that instead of solving a Lasso problem we solve a linear program (LP) to estimate the coefficient matrix. The steps of the algorithm are as follows. 
\begin{enumerate}
\item For each $i$ solve for $$ \arg \min \|\V{c}_{i}\|_1: \,\, \mathbf{I}_{\oi} \M{X}_{\oi} = \mathbf{I}_{\oi} \M{X}_{-i,\Omega} \V{c}_{i} $$ where $\M{X}_{\oi}$ denotes the data point $\M{X}_i$ with zeros filling at non-sampled locations and $\M{X}_{-i,\Omega}$ denotes the zero-filled data points except the $i$ data point. 
\item Collect the $\V{c}_i$ into a matrix $\M{C}$ and apply spectral clustering \cite{Luxburg:2007dq} to $\M{A} = |\M{C}| + |\M{C}|^\top$
\end{enumerate}

\section{Analysis of the algorithm}
\label{sec:analysis}

\subsection{Case I}

\begin{figure}[htbp]
\centering \makebox[0in]{
    \begin{tabular}{c c}
      \includegraphics[scale=0.4]{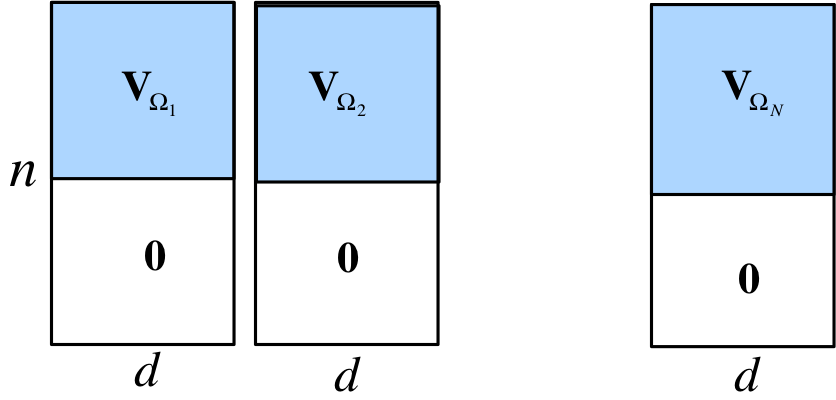}
 \end{tabular}}
  \caption{ Schematic depicting Case 1. All $\Omega_i = \Omega$.}
  \label{fig:Case1}
\end{figure}
Let $\Omega_{i} = \Omega$ for all $i$. In this case, all the points are sampled at the same locations. Let point $i$ belong to subspace $\ell$.  

Let $ \vell_{\Omega}= \mathbf{I}_{\Omega} \uell \in \mb{R}^{|\Omega| \times d}$ where $\mathbf{I}_{\Omega}$ denotes the projection onto the (non-zero) co-ordinates in $\Omega$. The following theorem summarizes the conditions under which SSC-EWZF algorithm results in correct clustering. 

\begin{theorem}
\label{thm:1}
Let $\oil = \Omega$ for all $i, \ell$ and $|\Omega| \geq d$.  Then SSC-LP leads to correct clustering if for all $ i \in [N_\ell]$, $k \neq \ell$, the following holds,
\begin{align}
\label{eq:case1}
\left| \frac{\V{\lambda}_{i}^{(\ell) \top}}{\| \lamil\|_2} (\vell_{\Omega})^{\dagger} \M{V}_{\Omega}^{(k)} \V{a}_{j}^{(k)} \right| < r_{in}(\mc{P}({\M{A}}_{-i}^{(\ell)}))\,\, ,
\end{align} 
where,
\begin{enumerate}
\setlength \itemsep{5pt}
\item $\lamil \in \arg \max_{\V{\lambda}} \langle \V{a}_{i}^{(\ell)}, \V{\lambda} \rangle\,\, : \,\, \| \M{A}_{-i}^{(\ell) \top} \V{\lambda}\|_{\infty} \leq 1$
\item $\mc{P}(\Aell_{-i}) = \{ \V{v}: \V{v} =  \Aell_{-i} \V{b}: \|\V{b}\|_{1} = 1\}$, and  $r_{in}(\mc{P}({\M{A}}_{-i}^{\ell}))$ denotes the in-radius of the centro-symmetric body. 
\end{enumerate}
\end{theorem}
\begin{proof}
The proof directly follows from the proof of [Theorem 6 in \cite{dim_red_sc_ext}].
\end{proof}

\textbf{Implications of the condition in Theorem \ref{thm:1}}: 

\begin{enumerate}
\setlength \itemsep{-2pt}
\item We note that the RHS in Equation (\ref{eq:case1}) is the same as the RHS for the full observation case, see \cite{soltanolkotabi2012}. Therefore, the performance degradation in clustering comes from increase in the LHS on an average under missing data.
\item The increase in RHS depends on how badly conditioned $\M{V}_{\Omega}^{(\ell)}$ is. If some of the singular values of $\M{V}_{\Omega}^{(\ell)}$ are very small, it will make the LHS in Equation \ref{eq:case1} large. To further understand this, let us consider a \emph{semi-random} model where in each subspace the data points are generated by choosing $\V{a}_{i}^{(\ell)}$ uniformly randomly on a unit sphere \cite{soltanolkotabi2012}. In this case it is easy to show that $\dfrac{\V{\lambda}_{i}^{(\ell) \top}}{\| \lamil\|_2}$ are also uniformly distributed on the unit sphere \cite{dim_red_sc_ext} and the \emph{expected value} of the LHS becomes $ \dfrac{\|(\vell_{\Omega})^{\dagger} \M{V}_{\Omega}^{(k)}\|_{F}}{d}$. This is essentially the (unnormalized) co-ordinate restricted coherence between subspaces $\ell$ and $k$. 
\item If the subspaces are sufficiently incoherent with the standard basis or satisfy an RIP like property for large enough $|\Omega|$, then the condition number of $\M{V}_{\Omega}^{(\ell)}$ is controlled and one can expect to obtain similar performance as the full observation case.

\item Note that \emph{while in this setting one can ensure perfect clustering, one cannot ensure either perfect subspace identification or completion}. This is because in this case the deterministic necessary conditions for identification and completion \cite{2014arXiv1410.0633P} are not satisfied. This indicates that clustering is an easier problem compared to subspace identification and completion under missing data.
\end{enumerate}

Here we would like to mention that the notion of in-radius is related to the notion of \textbf{permeance} \cite{Lerman2012ArXiv} of data points in a given subspace, quantifying how well the data is distributed inside each subspace. In-radius can be thought of as a worst-case permeance that doesn't scale with the number of data points, while permeance scales with the number of data points and is more of an averaged criteria. Perhaps this is the reason that the primal-dual analysis of SSC under full observation is not able to support the empirical evidence that as the number of points per subspace increases the clustering error goes down dramatically. For subspace clustering such the effect of the number of data points was shown more explicitly in a recent paper \cite{GSC2014}. A connection between these two quantities, namely the in-radius and permeance for subspace clustering under missing data will be undertaken in a future work.\\
We will analyze the more general case next. Analysis of the general case is hard and to gain \emph{insights} we break it up into several results. For this we introduce the following notation. 

\begin{figure}[htbp]
\centering \makebox[0in]{
    \begin{tabular}{c c}
      \includegraphics[scale=0.4]{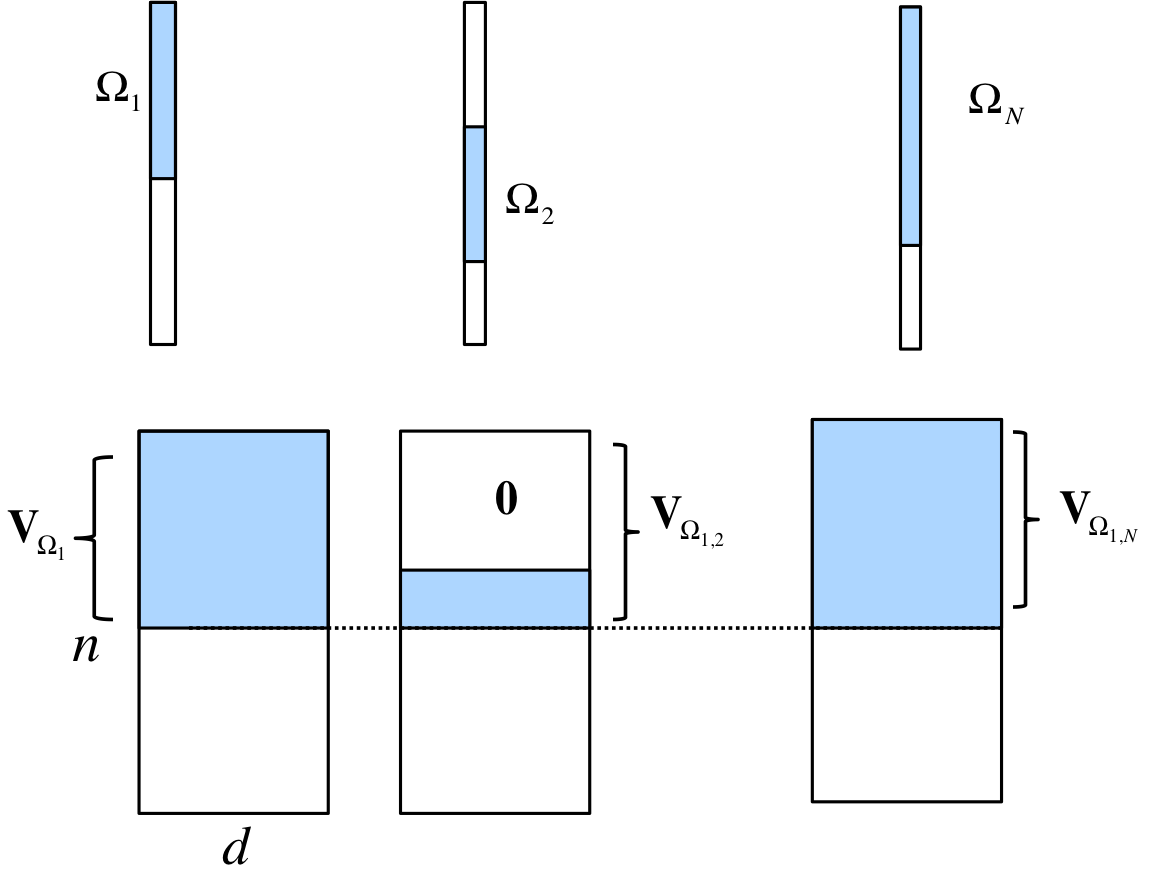}
 \end{tabular}}
  \caption{ Schematic depicting the case when all data is sampled at different locations and the notation for the basis $\M{V}_{\Omega_{i,j}}$ zero-filled and restricted to $\Omega_i$. Here we show for the case when $i =1$.}
  \label{fig:Case2}
\end{figure}

\textbf{Notation}: 

\begin{enumerate}
\setlength \itemsep{3pt}
\item Let $\M{V}_{\oi}^{(\ell)} = \M{I}_{\oil} \M{U}^{(\ell)}$. For other data points, $\M{X}_j$ if $\M{X}_j$ belongs to subspace $k \in \{1,2,..,L\}$, $\M{V}_{\Omega_{i,j}^{(k)}} = \M{I}_{\oil} \M{U}^{(k)}$ if $\Omega_{i}^{(\ell)} \subseteq \Omega_{j}^{(k)}$, else $\M{V}_{\Omega_{i,j}^{(k)}} = \M{I}_{\oil \cap \Omega_{j}^{(k)}} \M{U}^{(k)}$, where $\M{I}_{\oil \cap \Omega_{j}^{(k)}}$ is an $|\oil| \times |\oil|$ diagonal matrix with diagonal entries $1$ at locations in   $\M{I}_{\oil \cap \Omega_{j}^{(k)}}$ and $0$ otherwise. This is depicted in Figure, \ref{fig:Case2}.
\item Let $\M{V}_{\oi}^{(\ell)} = \M{Q}_{i}^{(\ell)} \Sigma_{i}^{(\ell)} \M{R}_{i}^{(\ell) \top}$ and $\M{V}_{\Omega_{i,j}}^{(\ell)} = \M{Q}_{i,j}^{(\ell)} \Sigma_{i,j}^{(\ell)} \M{R}_{i,j}^{(\ell) \top}$ be the SVDs of the truncated basis matrices. \\
We will refer to $\M{Q}_{i}^{(\ell)},\M{Q}_{i,j}^{(\ell)}$ as corresponding to \emph{co-ordinate restricted subspaces}.

\item Let \begin{align}
\label{eq:hata1}
\tilde{\V{a}}_{j}^{(\ell)} &= \M{Q}_{i}^{(\ell) \top} \M{V}_{\oji}^{(\ell)} \aell_{j} \,\,\forall j \\
 & = \M{Q}_{i}^{(\ell) \top} \M{Q}_{j}^{(\ell)} \underset{\hat{\V{a}}_{i,j}^{(\ell)}}{\underbrace{\Sigma_{i,j}^{(\ell)} \M{R}_{i,j}^{(\ell) \top} \aell_{j}}}
\end{align} and let $\tilde{\M{A}}_{-i}^{(\ell)}$ be the $r \times (N_\ell-1)$ matrix with columns as $\tilde{\V{a}}_{j}^{(\ell)}\,\, , j \neq i$.
\end{enumerate}

\subsection{Case II}

First we consider the case when $|\oi| = d$, i.e. we want to understand, if a given point is sampled at only $d$ locations, when can SSC-LP correctly assign it to the cluster of points from the same subspace.  Without loss of generality let the point come from the subspace $\ell$.

\textbf{Assumption}: For all possible supports $\Omega_{i}^{(\ell)}$, we assume that  $\M{Q}_{i}^{(\ell)} \in \mb{R}^{d \times d}$ is invertible. This assumption will be satisfied if the columns of $\M{U}^{(\ell)}$ are sufficiently incoherent with the standard basis. Under this assumption, we have the following theorem.
%

\begin{theorem}
\label{thm:2}
 Let $|\oil| = d$.  Then SSC-LP correctly assigns it to the data points coming from subspace $\ell$, if for all $k \neq \ell$ and $j$ , the following holds,
\begin{align}
\label{eq:case2}
\left| \frac{\V{\lambda}_{i}^{(\ell) \top}}{\| \lamil\|_2} \M{Q}_{i}^{(\ell)\top} \M{V}_{\Omega_{i,j}}^{(k)} \V{a}_{j}^{(k)} \right| < r_{in}(\mc{P}(\tilde{\M{A}}_{-i}^{(\ell)}))\,\, ,
\end{align} 
where,
\begin{enumerate}
\setlength \itemsep {5pt}
\item $\lamil \in \arg \max_{\V{\lambda}} \langle \hat{\V{a}}_{i}^{(\ell)}, \V{\lambda} \rangle\,\, : \,\, \| \tilde{\M{A}}_{-i}^{(\ell) \top} \V{\lambda}\|_{\infty} \leq 1 $
\item $\mc{P}(\tAell_{-i}) = \{ \V{v}: \V{v} =  \tAell_{-i} \V{b}: \|\V{b}\|_{1} = 1\}$, and  $r_{in}(\mc{P}(\tilde{\M{A}}_{-i}^{\ell}))$ denotes the in-radius of the centro-symmetric body. 
\end{enumerate}
\end{theorem}

\begin{proof}
The proof follows from slight modifications of the the arguments of the proof of Theorem \ref{thm:1}, which in turn follows from the proof of [Theorem 6 in \cite{dim_red_sc_ext}].
\end{proof}

\textbf{Implications of condition in Theorem \ref{thm:2}}:
\begin{itemize}
\setlength \itemsep{-2pt}

\item The in-radius $r(\mc{P}^{\circ}(\tilde{\M{A}}_{-i}^{(\ell)}))$ now \emph{depends on the sampling pattern within the subspace and it is different for different set of points}. Specifically note that it depends on the \emph{intrinsic subspace coherence} between the sampled points within the same subspace. In general it can be very small if the number of points $\M{X}_{j}^{(\ell)}$ (within the same subspace) that have good amount of overlap with $\M{X}_{i}^{(\ell)}$ is small.



\item The subspace coherence also depends on the \emph{relative} sampling pattern of points in other subspaces with respect to the point under consideration. This is reflected in the coherence term $\M{Q}_{i}^{(\ell)\top} \M{V}_{\Omega_{i,j}}^{(k)}$. Unlike Case I, for the semi-random model the distribution of the normalized dual direction is not uniform and we cannot comment on the average case performance. However, a worst-case condition can be obtained here from Equation \ref{eq:case2} which says that if \begin{align} \label{eq:case2a1}
\|\M{Q}_{i}^{(\ell)\top} \M{V}_{\Omega_{i,j}}^{(k)}\|_{2} \leq r_{in}(\mc{P}^{\circ}(\tilde{\M{A}}_{-i}^{(\ell)}))\end{align} then SSC-LP leads to correct clustering. This ofcourse requires that the co-ordinate projected subspaces be sufficiently incoherent as well as \emph{disjoint}.
\end{itemize}

\vspace{-3mm}
\subsection{Case III}

With these \textbf{insights} let us now proceed to handle the more general case when $|\oil| > d$. In order to do that let us first provide a geometric result using Lemma \ref{lem:1} (see Appendix) and then find a useful characterization of the dual certificate.\\

Again let  $\M{V}_{\Omega_i}^{(\ell)} = \M{Q}_{i}^{(\ell)} \Sigma_{i}^{(\ell)} \M{R}_{i}^{(\ell)\top}$ be the SVD of $\M{V}_{\Omega_i}^{(\ell)}$. Let $\M{x}_{\oil,j}$ denote the zero filled $\M{I}_{\Omega_j}\M{X}_j$ restricted to the indices in $\oil$ for all $j \in [N]$. Let $\M{X}_{-i,\oil}$ denote the matrix of these points except point $i$.

 Then SSC-LP solves $ \arg \min \|\V{c}_{i}\|_1: \,\, \M{X}_{i,\oil} = \M{X}_{-i,\oil} \V{c}_{i} $. Let $\V{c}_{i}^{(\ell)}$ and $\V{\nu}_{i}^{(\ell)}$ be solutions to,
$$ \arg \min \|\V{b}\|_1: \,\, \M{X}_{i,\oil}^{(\ell)} = \M{X}_{-i,\oil}^{(\ell)} \V{b}$$ 
$$ \arg \max \langle \M{X}_{i,\oil}^{(\ell)}, \V{\nu} \rangle: \,\, \| \M{X}_{-i,\oil}^{(\ell)\top} \V{\nu}\|_\infty \leq 1 $$ 

\noindent Assuming that \emph{strong duality} holds, the vectors $$\V{c} = [ \V{0},\cdots, \V{0}, \V{c}_{i}^{(\ell)},\cdots, \V{0}]^\top$$ and $\V{\nu}_{i}^{(\ell)}$ satisfy the conditions of Lemma \ref{lem:1} (See Appendix), hence guaranteeing correct subspace recovery if for all $k \neq \ell, j$, the following condition is satisfied. 
\begin{align} | \V{\nu}_{i}^{(\ell)\top} \M{X}_{\oil,j}^{(k)}| < 1 \end{align}

As such this condition does not provide any further intuition or insights. In order to arrive at results that are similar in flavor to Theorems \ref{thm:1} and \ref{thm:2}, let us further analyze this condition. We have the following \textbf{Assumption}: The matrix $\M{Q}_{i}^{(\ell)}$ has full column rank. This assumption is satisfied if columns of $\M{U}^{(\ell)}$ are sufficiently incoherent with respect to the co-ordinate basis.

Let $\M{Q}_{i}^{(\ell)} \M{Q}_{i}^{(\ell)\top}$ denote the projection matrix that projects a vector onto the subspace spanned by columns of $\qil \qilt$ and $(\M{I} - \qil \qilt)$ denote the null space of $\qilt$. Note that, $$\nuil =  \qil \qilt \nuil + (\M{I} - \qil \qilt)\nuil .$$ 
Then the condition for correct clustering becomes,
\begin{align} | ( \qil \qilt \nuil + (\M{I} - \qil \qilt)\nuil)^{\top} \M{X}_{\oil,j}^{(k)}| \leq 1 \end{align}
By triangle inequality if for some $0\leq \alpha \leq 1$, 
\begin{align}
\label{eq:case3a1} | (\nuil)^{\top} \M{Q}\M{Q}^{\top} \M{X}_{\oil,j}^{(k)}| < \alpha \end{align} 
and
\begin{align}
\label{eq:case3b1}
| (\nuil)^{\top} ( \M{I} - \M{Q}\M{Q}^{\top}) \M{X}_{\oil,j}^{(k)}| \leq (1-\alpha)
\end{align}
then the point $\M{X}_{i}^{(\ell)}$ is correctly clustered. Based on this we have the following theorem.

\begin{theorem}
\label{thm:3}
If for any point $\M{x}_i$ belonging to subspace $\ell$, the following conditions are satisfied,
\begin{align}
\label{eq:case3a}
\|\M{Q}_{i}^{(\ell)\top} \M{V}_{\Omega_{i,j}}^{(k)}\|_{2} < \alpha r_{in}\left(\mc{P}(\M{X}_{-i,\oil}^{(\ell)\top})\right)
\end{align}
\begin{align}
\|(\M{I} - \qil \qilt) \M{V}_{\Omega_{i,j}}^{(k)}\|_{2} \leq (1- \alpha) r_{in}\left(\mc{P}(\M{X}_{-i,\oil}^{(\ell)})\right)
\end{align}
for some $0 \leq \alpha \leq 1$, then SSC-LP succeeds. Here $\mc{P}(\M{X}_{-i,\oil}^{(\ell)}) = \{ \V{v}: \V{v} =  \M{X}_{-i,\oil}^{(\ell)} \V{b}: \|\V{b}\|_{1} = 1\}$, and  $r_{in}(\mc{P}(\M{X}_{-i,\oil}^{(\ell)}))$ denotes the in-radius of the centro-symmetric body
\end{theorem}
We note the following points. 
\begin{enumerate}
\item Compared to Theorems \ref{thm:1} and \ref{thm:2}, the conditions stated in Theorem \ref{thm:3} are worst-case and slightly weaker in the sense that they require the co-ordinate projected  subspaces to be disjoint.

\item Also note that for the most general case it is required that \emph{co-ordinate restricted subspace coherence} be small BUT ALSO that the \emph{projection error of the other set of points} onto the co-ordinate restricted subspace for the point under consideration be small. When $|\oil|$ is large, Equation \ref{eq:case3a} is more likely to be satisfied compared to Equation \ref{eq:case2a1}.

\item In the most general case, while more observations (per vector) is required, in this case under correct clustering if the sampling patterns and the number of data points satisfy the necessary and sufficient conditions in \cite{2014arXiv1410.0633P,2015arXiv150302596P} then it is possible to also correctly identify the subspaces and also ensure completion.
\end{enumerate}

\begin{figure}[htbp]
\centering
\vspace{-25mm}
\includegraphics[ width = .35\textwidth]{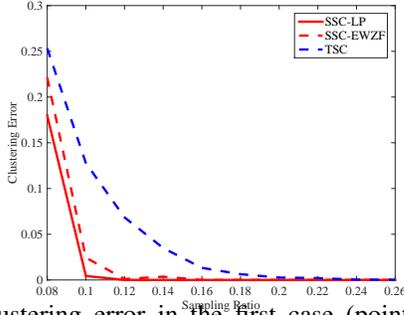} 
 \vspace{-23mm}
\caption{Clustering error in the first case (points sampled at same $pn$ co-ordinates)  for varying $p$.}
\end{figure}

\begin{figure*}[t!]
	\centering
	\begin{subfigure}[b]{0.32\textwidth}
		\centering
	\includegraphics[trim=.45in 1.5in .8in 2.8in, clip, width = \textwidth]{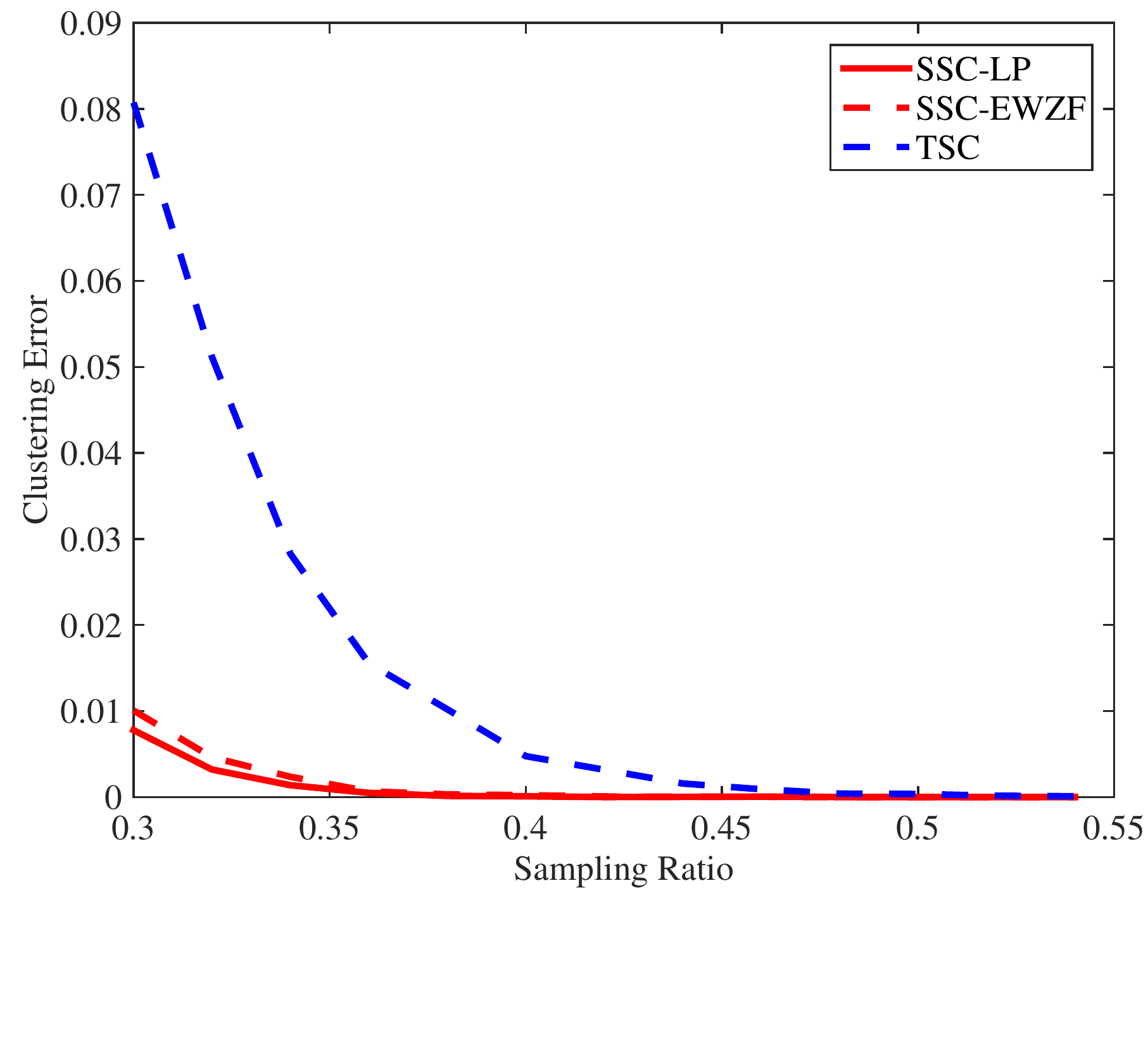}\vspace{-.3in}
		\caption{Clustering Error}
	\end{subfigure}%
	~ 
	\begin{subfigure}[b]{0.32\textwidth}
		\centering
	\includegraphics[trim=.6in 1.5in .8in 2.8in, clip, width = \textwidth]{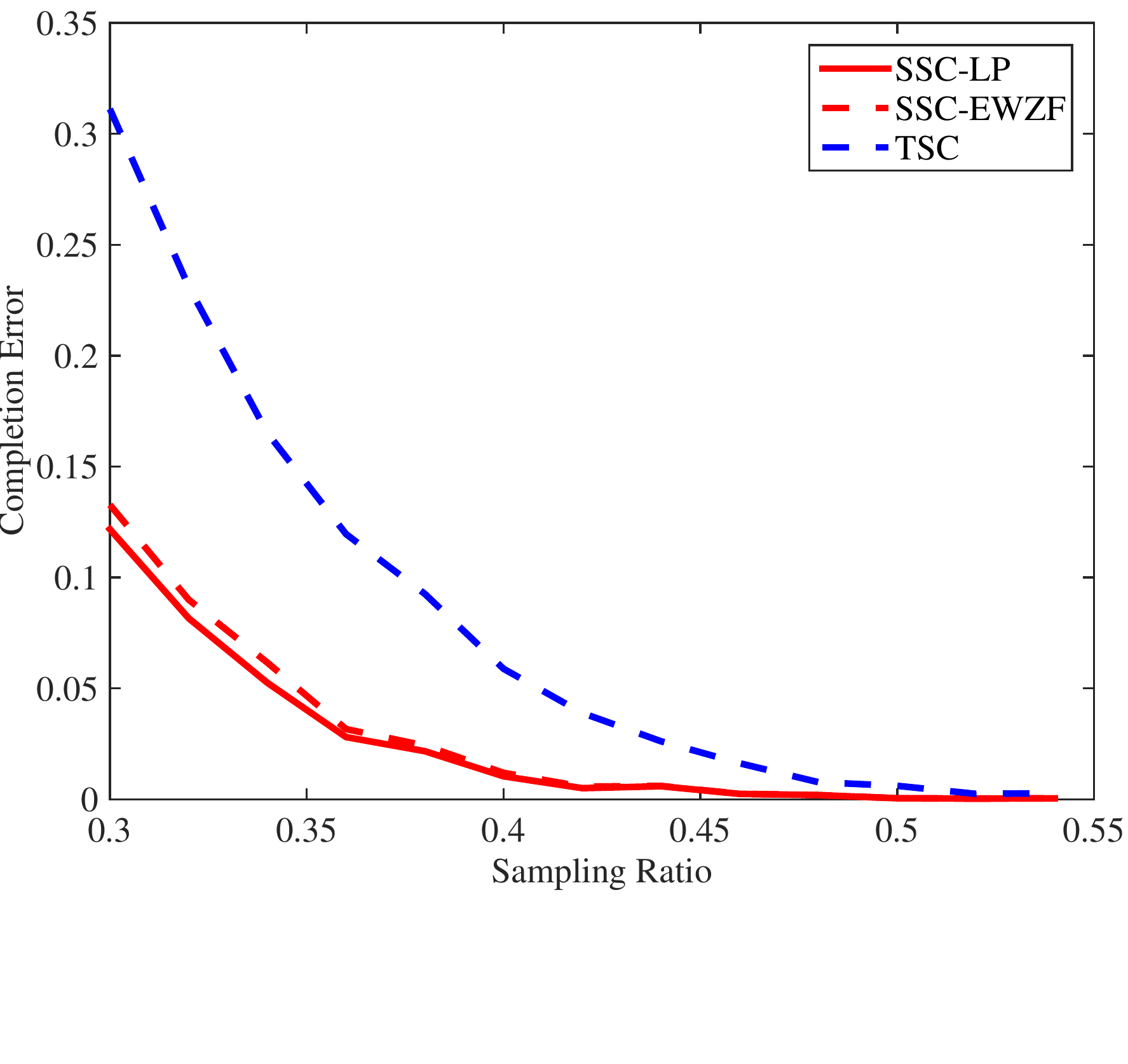}\vspace{-.3in}
		\caption{Completion Error}
	\end{subfigure}
		~ 
		\begin{subfigure}[b]{0.32\textwidth}
			\centering
			      \includegraphics[trim=.5in 1.5in .8in 2.8in, clip, width = \textwidth]{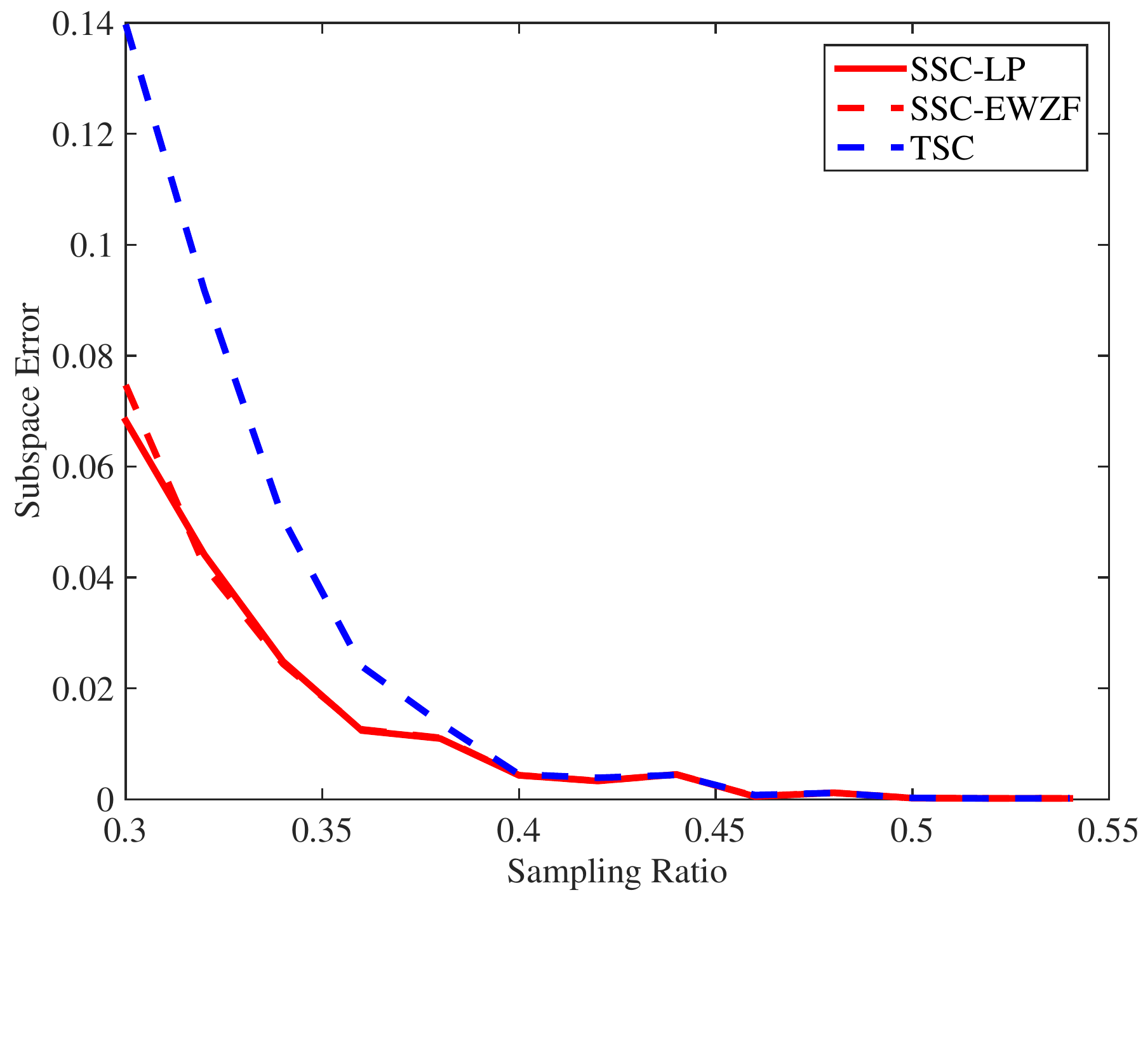}\vspace{-.3in}
			      \caption{Subspace Error}
		\end{subfigure}
	\caption{Different error metrics  for the case of random sampling.} \label{fig:tensor}\vspace{-.1in}
\end{figure*}

\vspace{-10mm}

\section{Numerical Results}

In this section, we will see the numerical performance of the proposed algorithm, SSC-LP, as compared to two other baseline algorithms for data clustering and completion under the UoS model with missing data. The data is generated by using $L$ rank $d$ subspaces, each composed of $N_l$ vectors of dimension $n$. We assume $n=50,\  L=3,\   d=3,$ and $N_l=150$ for the numerical results. Data in each subspace is generated by a multiplication of standard entry-wise Gaussian distributed $n\times d$ matrix and a standard entry-wise Gaussian distributed $d\times N_l$ matrix. 

We compare our algorithm with two other algorithms. The first is SSC-EWZF \cite{Yang2015}, which solves the Lasso problem rather than the linear program in this paper. This algorithm is selected as it gives the lowest clustering error among the different algorithms considered in \cite{Yang2015} (e.g., SSC-EC, SSC-CEC, ZF-SSC, MC-SSC).  The performance of SSC-EWZF algorithm largely depends on the choice of $\lambda$, which is chosen to be
\begin{align}
\lambda =\frac{\alpha}{\max_{i\neq j} |X_{\Omega_i}^\top X_{\Omega _j}|_{ij}},
\end{align}
where $\alpha $ is a tuning parameter\cite{Yang2015},  set to be 7.34 for our experiment (  selected by performing an optimized performance for different values of $\alpha$). 
The second is TSC algorithm proposed in \cite{tsc2014}, which builds adjacency matrix by thresholding correlations.  The thresholding parameter $q$ is given as
\begin{align}
q=\sqrt{N_l\log (N_l) },
\end{align}
such that $q$ is of an order smaller than  $N_l$ and larger than $\log N_l$  \cite{tsc2014}.
Since our proposed algorithm considers only observed entries when building adjacency matrix for each data, TSC algorithm could be a good comparison in respect to the way of building adjacency matrix. 

In our simulations, we consider two cases. The first case illustrates Case 1 in this paper,  where all the all the points are sampled at the same co-ordinates which are the  first $pn$ co-ordinates. 
The second case corresponds to Case 3 in this paper,  where missing data in each point is randomly sampled at a rounded value of $pn$ co-ordinates thus giving a sampling rate of approximately $p$ and a missing rate of approximately $1-p$. All the results are averaged over $100$ runs for the choice of the data and the sampled elements.

The comparisons for data clustering and  completion are performed using three metrics as explained further. The first metric  is the clustering error.  Clustering error is the ratio calculated by the number of wrongly classified data divided by the number of total data, same as defined in \cite{Yang2015}.
The clustering error for different sampling ratio $p$ in the first case is shown in Fig. 3, where SSC-LP shows the best clustering accuracy for all sampling ratios from 8\% to 26\%. Furthermore, the plot indicates that SSC-LP clusters data perfectly with $p=0.1$, equivalently at $5$ observations out of $50$ entries, while SSC-EWZF and TSC require $12$ and $11$ observations, respectively. Since the rank of each cluster is $3$ and only $5$ observations for each point are needed, we see that observing data at same co-ordinates need much less data for efficient clustering. We note that we cannot identify the subspace or complete the data with these observations further illustrating that clustering requires less number of observations than that required for data completion.

The clustering error for random sampling, the scenario described in the second case, with sampling ratio from 0.25 to 0.95 is shown in Fig 4(a), 
where we note that clustering error with our proposed algorithm, SSC-LP,  is the minimum among the three algorithms. Furthermore, the plot shows that the  sampling ratio at which the clustering error hits zero for the algorithms SSC-LP, SSC-EWZF and TSC are 0.38, 0.42, and 0.54, respectively. Thus, the proposed algorithm required least number of observations to efficiently cluster the data. We further note that the amount of data needed to cluster efficiently for random sampling (38\%) is larger than that for observing data at the same co-ordinates (10\%).  

The second metric is the completion error. Let the recovered matrix using a clustering algorithm be the output of matrix completion using SVT method 
\cite{candes2009exact} on the subspaces found as a result of the subspace clustering and the true matrix be the ground truth of the matrix with missing data. Then the recovery difference is defined as the matrix difference between recovered matrix and true matrix. Thus the completion error is measured by ratio of the Frobenius Norm of the recovery difference to the Frobenius Norm of the true matrix. 
The completion error for different values of $p$ can be seen  in Fig 4(b), 
 where we see completion error is positively correlated with clustering error and small percentage clustering error can result in large percentage completion error. Similar to the clustering error, SSC-LP has the lowest completion error among the three algorithms 
and the completion errors for SSC-LP, SSC-EWZF and TSC becomes zero at sampling ratio of around 0.50, 0.50, and 0.54 respectively, which are larger than the corresponding thresholds for the clustering errors.  

The third metric is the subspace error.  Subspace of a matrix with missing data can be recovered by finding the orthonormal basis of the recovered matrix. 
Projection difference is calculated by minus the projection of orthonormal basis of the recovered matrix from the orthonormal basis of the true matrix, and consequently, angle based subspace error is measured by $\arcsin$ of the normal of the projection difference, same as defined in 
\cite{bjorck1973numerical}. Mathematically, subspace error is expressed as:
\begin{displaymath}
\theta=\arcsin (\|(\M{B}-\M{A}\M{A}^\top\M{B})\|_{2})
\end{displaymath}
where $\theta$ is the angle based subspace error, $\M{A}$ is the orthonormal basis of the completed matrix and $\M{B}$ is the orthonormal basis of the true matrix.
With the simulation result in Fig 4(c), that subspace errors for all three algorithms start at 46\% sampling ratio.  For any sampling ratio lower than the threshold, subspace error for our proposed algorithm is the smallest among the three algorithms.
\vspace{-3mm}
\section{Conclusions}
This paper proposes an algorithm for sparse subspace clustering under missing data,  when data is assumed to come from a Union of Subspaces
(UoS) model, using a linear programming method, SSC-LP. Deterministic analysis of sufficient conditions when this algoirthm leads to correct clustering is presented. Extensive set of
simulation results for clustering as well as completion of data
under missing entries, under the UoS model are provided which demonstrate the effectiveness of the algorithm, and demonstrate that accurate
clustering does not imply accurate subspace identification.

In future we will derive the performance bounds in terms of average cases analysis of the deterministic conditions. In particular we will determine how the in-radius changes under element-wise sampling and how does the un-normalized subspace coherence behaves under these models.

\section{Appendix}
The following Lemma appears in \cite{soltanolkotabi2012}.
\begin{lemma}
\label{lem:1}
Let $\M{A}\in \mb{R}^{n \times N}$ and $\V{y} \in \mb{R}^{n}$ be given. If there exits $\V{c}$ obeying $\V{y} = \M{A} \V{c}$ with support $S \subset T$ and a dual certificate $\V{\nu}$ satisfying $$ \M{A}_{S}^{\top}\V{\nu} = \mbox{sign}(\V{c}_{S}), \,\, \|   \M{A}_{T\cap S^c}^{\top}\V{\nu} \|_{\infty}\leq 1, \,\, \|\M{A}_{T^c}^{\top}\V{\nu}\|_{\infty} < 1	$$ then all optimal solutions $\V{c}^*$ to $\arg \min_{\V{c}} \| \V{c}\|_{1}: \M{A} \V{y} = \V{c}$ satisfy $\V{c}_{T^c}^{*} = \mathbf{0}$.  
\end{lemma}

\vspace{-3mm}
\bibliographystyle{IEEEbib}
\bibliography{SSmC_Geom,SSmCbibliography,NewBib}

\begin{thebibliography}{10}

\bibitem{soltanolkotabi2012}
Mahdi Soltanolkotabi and Emmanuel~J. Candes,
\newblock ``A geometric analysis of subspace clustering with outliers,''
\newblock {\em Ann. Statist.}, vol. 40, no. 4, pp. 2195--2238, 08 2012.

\bibitem{Elhamifar:2012uz}
E~Elhamifar and R~Vidal,
\newblock ``{Sparse Subspace Clustering: Algorithm, Theory, and
  Applications},''
\newblock {\em Pattern Analysis and Machine Intelligence, IEEE Transactions
  on}, vol. 35, no. 11, pp. 2765--2781, 2013.

\bibitem{bal2}
L.~Balzano, A.~Szlam, B.~Recht, and R.~Nowak,
\newblock ``K-subspaces with missing data,''
\newblock in {\em Statistical Signal Processing Workshop (SSP), 2012 IEEE}, Aug
  2012, pp. 612--615.

\bibitem{bal3}
D.~Pimentel, R.~Nowak, and L.~Balzano,
\newblock ``On the sample complexity of subspace clustering with missing
  data,''
\newblock in {\em Statistical Signal Processing (SSP), 2014 IEEE Workshop on},
  June 2014, pp. 280--283.

\bibitem{ErikssonArXiv2011}
Brian Eriksson, Laura Balzano, and Robert~D. Nowak,
\newblock ``High-rank matrix completion and subspace clustering with missing
  data,''
\newblock {\em CoRR}, vol. abs/1112.5629, 2011.

\bibitem{Yang2015}
Congyuan Yang, Daniel Robinson, and Rene Vidal,
\newblock ``Sparse subspace clustering with missing entries,''
\newblock in {\em Proceedings of The 32nd International Conference on Machine
  Learning}, 2015, pp. 2463?Äì--2472.

\bibitem{Luxburg:2007dq}
Ulrike Luxburg,
\newblock ``A tutorial on spectral clustering,''
\newblock vol. 17, no. 4, pp. 395--416, 2007.

\bibitem{dim_red_sc_ext}
Reinhard Heckel, Michael Tschannen, and Helmut B\"olcskei,
\newblock ``Dimensionality-reduced subspace clustering,''
\newblock {\em Journal of Machine Learning Research}, July 2015.

\bibitem{2014arXiv1410.0633P}
D.~L. {Pimentel-Alarc{\'o}n}, R.~D. {Nowak}, and N.~{Boston},
\newblock ``{Deterministic Conditions for Subspace Identifiability from
  Incomplete Sampling},''
\newblock {\em ArXiv e-prints}, Oct. 2014.

\bibitem{Lerman2012ArXiv}
Gilad Lerman, Michael~B. McCoy, Joel~A. Tropp, and Teng Zhang,
\newblock ``Robust computation of linear models, or how to find a needle in a
  haystack,''
\newblock {\em CoRR}, vol. abs/1202.4044, 2012.

\bibitem{GSC2014}
D.~{Park}, C.~{Caramanis}, and S.~{Sanghavi},
\newblock ``{Greedy Subspace Clustering},''
\newblock {\em ArXiv e-prints}, Oct. 2014.

\bibitem{2015arXiv150302596P}
D.~L. {Pimentel-Alarc{\'o}n}, N.~{Boston}, and R.~D. {Nowak},
\newblock ``{A Characterization of Deterministic Sampling Patterns for Low-Rank
  Matrix Completion},''
\newblock {\em ArXiv e-prints}, Mar. 2015.

\bibitem{tsc2014}
Reinhard Heckel and Helmut B\"olcskei,
\newblock ``Robust subspace clustering via thresholding,''
\newblock {\em IEEE Transactions on Information Theory}, 2015.

\bibitem{candes2009exact}
Emmanuel~J Cand{\`e}s and Benjamin Recht,
\newblock ``Exact matrix completion via convex optimization,''
\newblock {\em Foundations of Computational mathematics}, vol. 9, no. 6, pp.
  717--772, 2009.

\bibitem{bjorck1973numerical}
Ake Bj{\"o}rck and Gene~H Golub,
\newblock ``Numerical methods for computing angles between linear subspaces,''
\newblock {\em Mathematics of computation}, vol. 27, no. 123, pp. 579--594,
  1973.

\end{thebibliography}

\end{document}